\pdfoutput=1
\documentclass{scrartcl}
\PassOptionsToPackage{createShortEnv,commandRef=cref}{proof-at-the-end}

\def\myshorttitle{On the (In)feasibility of ML Backdoor Detection as an Hypothesis Testing Problem}
\def\github{\href{https://github.com/g-pichler/in_feasibility_of_ml_backdoor_detection}{https://github.com/g-pichler/in\_feasibility\_of\_ml\_backdoor\_detection}}

\usepackage{graphicx} 
\usepackage{xcolor}
\usepackage{amsmath, amssymb, amsthm}
\usepackage{dsfont}
\usepackage{mathtools}
\usepackage{mathrsfs} 
\usepackage{glossaries-extra}
\usepackage{tikz}
\usepackage{expl3}
\usepackage{caption}
\usepackage{subcaption}
\usepackage{enumitem}
\usepackage{booktabs}
\usepackage{wrapfig}
\usepackage{changes}
\usepackage{comment}
\usepackage[hypertexnames=false]{hyperref}
\usepackage[capitalize]{cleveref}
\usepackage[createShortEnv,commandRef=cref]{proof-at-the-end}
\usetikzlibrary{decorations.pathreplacing,
  calligraphy,
  matrix,
  positioning,
  calc}

\title{\myshorttitle}
\author{Georg Pichler, Marco Romanelli, Divya Prakash Manivannan, \\Prashanth Krishnamurthy, Farshad Khorrami, Siddharth Garg}

\newcommand{\indicator}{\mathds{1}}
\newcommand{\ind}{\indicator}
\newcommand{\backdoor}{\ensuremath{b}}
\newcommand{\Exp}{\mathbb{E}}

\newcommand{\DD}{\mathcal{D}}
\newcommand{\DDi}{\mathcal{D}'}

\DeclareMathOperator{\TV}{TV}

\makeatletter
\ifdefined\supplementary
  \def\envsub{section}
\else
  \def\envsub{\@empty}
\fi
\makeatother

\newtheorem*{theorem*}{Theorem}
\newtheorem{corollary}{Corollary}[\envsub]
\newtheorem*{corollary*}{Corollary}

\newtheorem{lemma}{Lemma}[\envsub]
\newtheorem*{lemma*}{Lemma}

\newtheorem{definition}{Definition}[\envsub]

\theoremstyle{remark}
\newtheorem{remark}{Remark}[\envsub]
\newtheorem{example}{Example}[\envsub]

\crefname{type}{Type}{Types}
\Crefname{type}{Type}{Types}
\crefname{equation}{}{}
\Crefname{equation}{}{}

\makeatletter
\let\thmt@shortoptarg\@empty
\let\XR@prefix\@empty
\makeatother

\def\myalpha{0.1}
\def\mybeta{0.001}

\def\mysigma{0.5}
\def\myv{(0.981, 0.196)}
\def\mypval{0.2381}
\def\mygamma{0.5}
\def\myN{150}
\def\mymu{1.177}

\makeglossaries

\newabbreviation{mbd}{MBD}{Model Backdoor Detection}
\newabbreviation{sbd}{SBD}{Sample Backdoor Detection}
\newabbreviation{cbd}{CBD}{Combined Backdoor Detection}
\newabbreviation{ood}{OOD}{Out-Of-Distribution}
\newabbreviation{ml}{ML}{Machine Learning}
\newabbreviation{ann}{ANN}{Artificial Neural Network}
\newabbreviation{pac}{PAC}{Probably Asymptotically Correct}
\newabbreviation{ttsd}{TTSD}{Test-Time Trigger Sample Detection}
\newabbreviation{sgd}{SGD}{Stochastic Gradient Descent}

\allowdisplaybreaks

\definechangesauthor[name={Georg Pichler}, color=cyan]{GP}
\usepackage{natbib}
\usepackage{geometry}

\pratendSetGlobal{restate,end}

\begin{document}
\maketitle

\begin{abstract}
  We introduce a formal statistical definition for the problem of backdoor detection in machine learning systems and use it to analyze the feasibility of such problems, providing evidence for the utility and applicability of our definition. The main contributions of this work are an impossibility result and an achievability result for backdoor detection. We show a no-free-lunch theorem, proving that universal (adversary-unaware) backdoor detection is impossible, except for very small alphabet sizes. Thus, we argue, that backdoor detection methods need to be either explicitly, or implicitly adversary-aware. However, our work does not imply that backdoor detection cannot work in specific scenarios, as evidenced by successful backdoor detection methods in the scientific literature. Furthermore, we connect our definition to the probably approximately correct (PAC) learnability of the out-of-distribution detection problem. 
  
\end{abstract}

\section{\uppercase{Introduction}}
The adoption of modern \gls{ml} methods in a range of
real-world tasks including navigation~\cite{Chen2022Weakly,Wang2022Towards}, medical diagnosis~\cite{Varoquaux2022Machine,Tchango2022DDXPlus}, and system control~\cite{Zhang2023Energy} has grown dramatically. However, safe and trustworthy \gls{ml} systems remain elusive~\cite{Ilyas2019Adversarial,Wu2022BackdoorBench}, for reasons including 
poor interpretability~\cite{Burkart2021survey,Roscher2020Explainable},
test time adversarial inputs~\cite{Goodfellow2015Explaining}
and, relevant to this paper, \textit{training time poisoning and backdooring attacks}~\cite{Gu2017BadNets}.
As the scale, complexity and training data requirements of modern deep neural network architectures has grown, few can afford to train models from scratch.
Many users therefore download and fine-tune pre-trained models, or deploy them as is. Consequently, purposefully implanted backdoors in pre-trained \gls{ml} models pose a key security risk for future \gls{ml} deployments.

In the classic backdoor threat model, a malicious actor trains a backdoored \gls{ml} model by altering its training data.
During inference, for certain, \emph{backdoored} inputs, modified in an attacker chosen way, the model then provides erroneous predictions.
For example, ~\cite{Gu2017BadNets} demonstrate a backdoor in a traffic sign detector that misclassifies stop signs as speed limit signs when stop signs are modified with stickers or sticky notes. Here the sticker or sticky note serve as a trigger, misleading the model into making incorrect decisions. 
While there are many ways such a backdoor could be embedded into a model, 
prior work shows that altering even a small fraction of training data yields models with stealthy and effective backdoors~\cite{Qi2023Revisiting}.

In light of these attacks, substantial efforts have been devoted to \emph{backdoor defenses} with the goal of identifying such backdoor attacks. 
To detect a backdoor, the model user (i.e., the \emph{defender}) has access to a, typically small, validation dataset of clean inputs.
In the \gls{mbd} problem~\cite[Sec.~4.2]{Chen2019DeepInspect},~\cite{Wang2019Neural,Shen2021Backdoor}, the defender wishes to detect if the 
model itself contains a backdoor.
In the \gls{sbd} problem~\cite{Liu2023Detecting,Ma2022Beatrix}, 
the defender wants to detect if a specific test input is backdoored or not, assuming that models deployed in the field might be backdoored.
We note that our backdooring threat model is part of the larger body of work on data poisoning threats. The latter encompasses all scenarios where
a model trained on--partially--poisoned data is negatively impacted in some way, possibly, but necessarily, by implanting a backdoor.
With the growing use of large pretrained foundation models, the backdooring threat (where the defender receives a model, not training data) is of increasing relevance.

Unfortunately, despite several years of research, the field is still \emph{plagued 
by a cat-and-mouse game between attackers and defenders}. Certifiable defenses against backdooring attacks have remained elusive, and despite some recent progress in this direction for data poisoning attacks, those results don't translate to our setting as discussed in \cref{sec:related_works}.
In fact, despite the large body of 
work in the area, backdoor detection has \emph{not been formally studied} from first principles.

Here, we undertake the first \emph{formal} exposition of the \gls{ml} backdooring problem for both
the \gls{mbd} and \gls{sbd} settings. Although it has been observed that backdoor detection can work in specific scenarios, we are interested in the feasibility of detecting arbitrary, unknown backdoors. We thus ask the following questions. \emph{First}, what are fundamental bounds lower bounds on backdoor detection---in fact, is backdoor detection even feasible, and if so, under what assumptions? 
\emph{Second}, how is backdoor detection related to other statistical problems in \gls{ml}? Recent work has leveraged
\gls{ood} detection methods for backdoor detection; can this relationship be formalized?
\emph{Third}, how are \gls{mbd} and \gls{sbd} related? Both are separately addressed in literature, but their relationship 
has not been explicated. And \emph{finally}, what are the implications for future progress in building practical backdoor
defenses?
We provide answers to all four questions using theoretical results and a ``toy'' example.

 
\textbf{Contributions.}
In this paper, we present the first precise statistical formulation of the 
\gls{mbd} and \gls{sbd} problems (\cref{sec:threat-model}). This formulation enables several new insights
on backdoor detection. 
\begin{enumerate}
\item \emph{Relationship to well-known statistical problems:} Our formulation unifies \gls{mbd}, \gls{sbd} and even \gls{ood} detection within a common framework and we reduce these problems to standard statistical hypothesis testing problems.
\item \emph{Infeasibility:} Leveraging these reductions, we conclude that under realistic assumptions, universal (adversary-unaware) 
backdoor detection is not possible for an infinite alphabet of the training data.
\item \emph{Bound for finite alphabet size:} For a finite data alphabet, we provide a bound on the achievable error probability given a fixed training set size. These bounds are evaluated for commonly used datasets in \gls{ml}, showing that universal backdoor detection is only achievable for very small alphabets.
\item \emph{Connections to \gls{pac} learning theory of \gls{ood} detection:} We show that detecting a backdoor in training data is equivalent to a binary Neyman-Pearson hypothesis test if \gls{ood} detection is \gls{pac} learnable as defined in \cite{Fang2022Is}.
\item \emph{Methodological weakness in existing defenses:} we observe that almost all defense strategies only evaluate 
on backdoored models, and fail to report \emph{false positive} rates in the likely case that models are actually clean.
\end{enumerate}

\section{\uppercase{Theoretical Formulation and Results}}
\label{sec:results}

We focus on \gls{mbd} and \gls{sbd} as introduced above, in the case, where the attacker has limited control over the training data and is able to backdoor a certain portion of the dataset.
The training itself is performed using a standard method, e.g., \gls{sgd}.
For an extensive overview of other empirical backdoor problems, the reader is referred to, e.g., \cite{Wu2022BackdoorBench}. 

\subsection{Formulating \glsentrylong{mbd} (\glsentryshort{mbd})}
\label{sec:threat-model}

\textbf{Overview.}
After $N$ samples of training data are collected, the backdoor attacker has the option of backdooring a portion of the training data, by replacing each \emph{clean} sample with a \emph{backdoored} sample. This backdooring may alter, e.g., images as well as their labels. Subsequently, an \gls{ann} is trained on the resulting training set. Given the resulting trained network (i.e., the network parameters), the task of the backdoor detector is to determine whether the training data had been backdoored. The detector may obtain $M$ additional \emph{clean} samples, e.g., by independently collecting additional data. We assume that the backdoor attacker has no access to these samples. \Cref{tab:notation} provides an overview of the notation used.

\begin{table*}[h]
  \centering
  \begin{tabular}{ll}
    \toprule
    \textbf{Symbol} & \textbf{Description} \\
    \midrule
    $X \sim P$ & $X$ distributed according to $P$ \\
    $\mathcal X$ & Alphabet of $X$ \\
    $|\mathcal X| \in \mathbb N \cup \{\infty\}$ & Alphabet size \\
    $\mathcal P(\mathcal X)$ & Set of all probability distributions on $\mathcal X$ \\
    $X^{(0)} \sim P_0$ & Clean sample/distribution \\
    $X^{(1)} \sim P_{\backdoor}$ & Backdoored sample/distribution \\
    $\gamma \in [0,1]$ & Backdoor probability \\
    $P_1 = \gamma P_{\backdoor} + (1-\gamma) P_0$ & Mixture of $P_0$ and $P_{\backdoor}$ \\
    $\DD = (X_1, \dots X_N) \sim P^N$ & Dataset of $N$ i.i.d.\ samples \\
    $\DD^{(0)} = (X^{(0)}_1, \dots X^{(0)}_N) \sim P_0^N$ & Clean dataset of $N$ i.i.d.\ samples \\\
    $\DD^{(1)} = (X^{(1)}_1, \dots X^{(1)}_N) \sim P_1^N$ & Backdoored dataset of $N$ i.i.d.\ samples \\\
    $\DDi = (X'_1, \dots X'_M) \sim P_0^M$ & Dataset of $M$ additional clean samples \\
    $\theta = \mathcal A(\DD)$ & Parameters resulting from training on $\DD$ \\
    $\mathcal P \subseteq P(\mathcal X)^2$ & Set of allowed distributions $(P_0, P_{\backdoor}) \in \mathcal P$ \\ 
    $J \sim \mathcal B(\frac12)$ & Bernoulli random variable indicating a backdoor \\
    $\mathbf Q = (\mathcal A(\DD^{(J)}), \DDi)$ & Parameters $\theta$ and samples $\DDi$ given to the detector \\
    $g(\mathbf Q) \in \{0,1\}$ & Backdoor detector \\ 
    $R(g; P_0, P_{\backdoor}) = \Pr\{g(\mathbf Q) \neq J\}$ & Risk of detector $g$ at $(P_0, P_{\backdoor}$) \\
    $\alpha \in [0, \frac12]$ & Error probability \\
    $g_0 = g$, $g_1$, $g_2$, $g_3$ & Type 0,1,2,3 detectors \\
    $\mathbf{Q}_0 = \mathbf{Q}$, $\mathbf{Q}_1$, $\mathbf{Q}_2$, $\mathbf{Q}_3$ & Input for Type 0,1,2,3 detectors \\
    $\TV(P, Q)$ & Total variational distance between two distributions \\
    $\beta \in [0,1)$ & Distance constraint $\TV(P_0, P_{\backdoor}) \ge 1 - \beta$ \\
    \bottomrule
  \end{tabular}
  \caption{Overview of Notation. This excludes \cref{sec:sbd_and_ood}, where some definitions are generalized.}
  \label{tab:notation}
\end{table*}

\textbf{Dataset and training.} Consider a, possibly stochastic, training algorithm $\mathcal A$ (e.g., \gls{sgd}), that trains a model on training data\footnote{The training sample $X$ may be a vector that includes data and label.} $\DD = (X_1, X_2, \dots, X_N)$, consisting of $N$ i.i.d.\ random variables, distributed like $X \sim P$, as input and produces a parameter vector $\theta = \mathcal A(\DD)$ as output.

\textbf{Clean data.} Let $P_0 \in \mathcal P(\mathcal X$) be the probability distribution on $\mathcal X$ of clean samples and let $\DD^{(0)} = (X_1^{(0)}, X_2^{(0)}, \dots, X_N^{(0)})$ be a clean dataset, consisting of $N$ i.i.d.\ random variables, drawn from $P_0$. 

\textbf{Backdoor.}
To \emph{backdoor} a model, an adversary may replace some training samples with backdoored samples,  
drawn from a different distribution $P_{\backdoor} \in \mathcal P(\mathcal X)$.
This distribution may result from applying a backdoor function to the clean samples. Note, that as the training sample $X$ includes the data and the label, both may be altered by the adversary.

\textbf{Backdoored training data.} Assuming that a fraction $\gamma \in (0,1]$ of the training data is backdoored, 
the backdoored training dataset $\DD^{(1)} = (X_1^{(1)}, X_2^{(1)}, \dots, X_N^{(1)})$ is independently drawn according to $P_1 = \gamma P_{\backdoor} + (1-\gamma) P_0$, i.e., according to $P_{\backdoor}$ with probability $\gamma$ and from $P_{0}$ with probability $1-\gamma$.

\textbf{Additional clean data.}
Furthermore, let $\DDi = (X_1', X_2', \dots, X_M') \sim P_0^M$ be $M$ i.i.d.\ additional clean samples distributed according to $P_0$. These samples correspond to clean validation data or may have been collected by the backdoor detector prior to making a decision.

\textbf{\glsentrylong{mbd}.} The backdoor detector is a function $g$, that takes $\theta = \mathcal A(\DD^{(j)})$ and additional data $\DDi$ as its input and outputs $0$ for \emph{``backdoor''} and $1$ for \emph{``no backdoor''}.
For \gls{mbd}, we require the detector to determine $j$ with high probability. For ease of notation, we use a Bernoulli-$\frac12$ random variable $J \sim \mathcal B(\frac12)$ and define the input for the detector as $\mathbf Q = (\mathcal A(\DD^{(J)}), \DDi)$, such that the error probability $\Pr\{g(\mathbf Q) \neq J\}$ of the detector is well-defined.

\textbf{Possible data distributions.} Finally, the last observation needed to obtain a well-defined backdoor detection problem is, that we need to avoid the possibility of $P_0 = P_{\backdoor}$. In the case where the clean and the backdoor distributions are identical, clearly, detection is impossible. We opt for the general approach of defining a suitable set $\mathcal P \subseteq \mathcal P(\mathcal X)^2$ that contains all possible clean and backdoor distribution pairs $(P_0, P_{\backdoor}) \in \mathcal P$.

These discussions then naturally lead to the following central \lcnamecref{def:alpha-error}.

\begin{definition}
  \label{def:alpha-error}
  The \gls{mbd} problem for a training algorithm $\mathcal A$ is determined by the following quantities: 
  $\gamma \in (0,1]$, 
  $N \in \mathbb N$,
  $M \in \mathbb N$, and
  $\mathcal P \subseteq \mathcal P(\mathcal X)^2$.

  Fixing these quantities, we define the \emph{risk} of a backdoor detector $g$ associated with $(P_0, P_{\backdoor})$ as
  \begin{align}
    \label{eq:risk}
    R(g; P_0, P_{\backdoor}) &:= \Pr\{g(\mathbf Q) \neq J\} \\
                             &= \frac 12 \sum_{j=0,1} \Pr\{g(\mathcal A(\DD^{(j)}), \DDi) \neq j\} .
  \end{align}
  
  We say that a backdoor detector is \emph{$\alpha$-error} for some $\alpha \in [0,\frac 12]$ if, for every pair $(P_0, P_{\backdoor}) \in \mathcal P$, the risk is bounded by 
  \begin{align}
    \label{eq:alpha-error}
    R(g; P_0, P_{\backdoor}) \le \alpha .    
  \end{align}
\end{definition}

\begin{remark}
  \label{rmk:type_I_II}
  Instead of bounding the risk as in \cref{eq:alpha-error}, it may seem more natural to require $\Pr\{g(\mathbf Q) \neq j|J=j\} \le \alpha$ for $j=0,1$. But note that $\Pr\{g(\mathbf Q) \neq j|J=j\} \le 2\alpha$ for $j = 0,1$ immediately follows from \cref{eq:alpha-error}.
\end{remark}

\subsection{(In)feasibility of \glsentrylong{mbd}}
\label{sec:infeasibility}

It will be useful to consider easier problems than $\alpha$-error detection (\cref{def:alpha-error}) and establish reductions. To this end, we consider four different \emph{\namecrefs{itm:default}} of detectors, starting with \cref{itm:default} that corresponds to \gls{mbd}. The other three detectors also seek to infer $J$, but are given access to progressively \emph{more} information via oracles; as such, any subsequent detectors can improve on previous ones. The four detectors are: 
\begin{enumerate}[start=0,font=\textbf,label={Type \arabic*:},ref=\arabic*,labelindent=1em,itemindent=3em]
\item \label[type]{itm:default} The default detector, $g_0(\mathbf Q_0)$, as used in \cref{def:alpha-error} with with $\mathbf Q_0 = \mathbf Q = (\mathcal A(\DD^{(J)}), \DDi)$. This detector corresponds to our \gls{mbd} problem.
\item \label[type]{itm:harder} Detector $g_1(\mathbf Q_1)$ with $\mathbf Q_1 = (\DD^{(J)}, \DDi)$, i.e., with access to the training dataset $\DD^{(J)}$ instead of just the trained model, and $M$ independent clean validation samples $\DDi$. 
\item \label[type]{itm:hard} Detector $g_2(\mathbf Q_2)$ with $\mathbf Q_2 = (\DD^{(J)}, P_0)$, i.e., with access to the clean data distribution instead of clean validation samples. 
This is an \gls{ood} detection problem: Is $\DD^{(J)}$ \gls{ood} with respect to $P_0$?
\item \label[type]{itm:easier} Detector $g_3(\mathbf Q_3)$ with $\mathbf Q_3 = (\DD^{(J)}, P_0, P_{\backdoor})$; i.e., the previous detector also now gets the backdoor distribution $P_{\backdoor}$, and must decide which distribution the training data came from. This is classical binary Neyman-Pearson hypothesis testing problem between $P_0^N$ and $P_1^N$.
\end{enumerate}

We assume that detectors of \cref{itm:hard,itm:easier} have access to $P_0$ (and $P_{\backdoor}$ for a \cref{itm:easier} detector) in terms of evaluation of the distribution, and also have the ability to sample from the distribution. We thus consider \cref{itm:hard,itm:easier} as randomized detectors to account for sampling. The definitions of risk and $\alpha$-error detection of $g_2, g_3$ apply mutatis mutandis as in \cref{def:alpha-error}, where the probability in \cref{eq:risk} is also taken over the randomness of $g$.

\begin{remark}[Ordering of detector \namecrefs{itm:default}]
  \label{rmk:ordering}
  \Cref{itm:default,itm:harder,itm:hard,itm:easier} are listed in order of decreasing difficulty as, e.g., more information is provided to a \cref{itm:easier} detector than to a \cref{itm:hard} detector. Thus, an $\alpha$-error detector $g$ immediately provides an $\alpha$-error \cref{itm:harder} detector $g_1$, which in turn immediately provides an $\alpha$-error \cref{itm:hard} detector $g_2$, which yields an $\alpha$-error detector $g_3$ of \cref{itm:easier}.
  Thus, we can define a total ordering on the different \namecrefs{itm:default} of detectors, using $A \prec B$ to signify that $A$ can be derived from $B$: $\mathbf Q_0 \prec \mathbf Q_1 \prec \mathbf Q_2 \prec \mathbf Q_3$.
  The formal argument, showing this claim can be found in \cref{lem:reduction} in \cref{sec:additional-results}.
\end{remark}

In \cref{sec:impossibility} we will show that for a reasonable $\mathcal P$, $\alpha$-error \cref{itm:hard} detection is impossible with $\alpha < \frac12$. The reduction argument in \cref{rmk:ordering} thus ensures that $\alpha$-error detection with $\alpha < \frac12$ is also impossible for \cref{itm:default} and \cref{itm:harder} detectors.

We can resolve the situation for a \cref{itm:easier} detector using the Neyman-Pearson lemma.
\begin{lemmaE}
  \label{lem:easier}
  Given a \cref{itm:easier} backdoor detector $g_3(\DD, P_0, P_{\backdoor})$, for any pair $(P_0, P_{\backdoor}) \in \mathcal P(\mathcal X)^2$ we have
  \begin{align}
    R(g_3; P_0, P_{\backdoor}) &\ge \frac{1}{2} - \frac{1}{2} \TV(P_0^N, P_{1}^N) \label{eq:easier_sol} \\
                               &\ge \frac{1}{2} - \frac{\gamma N}{2} \TV(P_0, P_{\backdoor}),
  \end{align}
  where the first equality in \cref{eq:easier_sol} can be achieved by the Neyman-Pearson detector.
  Thus, an $\alpha$-error detector of \cref{itm:easier} can only exist if $\alpha \ge \frac{1}{2} - \frac{\gamma N}{2} \TV(P_0, P_{\backdoor})$ for all $(P_0, P_{\backdoor}) \in \mathcal P$.
\end{lemmaE}
\begin{proofE}
  Fix $(P_0, P_{\backdoor})$ and let $\mathcal Q = \{\mathbf{x} \in \mathcal X^N: g_3(\mathbf x, P_0, P_{\backdoor}) = 1\}$ to obtain
  \begin{align}
    & 1 - R(g_3; P_0, P_{\backdoor}) = \frac12 \sum_{j\in\{0,1\}} \Pr\{g_3(\DD^{(j)}, P_0, P_{\backdoor}) = j\} \nonumber\\
                                   &\qquad= \frac{1}{2} \int \ind_{\mathcal Q} \, dP_{1}^N + \frac{1}{2} \int \ind_{\mathcal Q^c}  \,dP_{0}^N \\
                                   &\qquad= \frac{1}{2} + \frac{1}{2} \int \ind_{\mathcal Q} \, dP_{1}^N - \frac{1}{2} \int \ind_{\mathcal Q}  \,dP_{0}^N \\
                                   &\qquad= \frac{1}{2} + \frac{1}{2} \int \ind_{\mathcal Q} \, d(P_{1}^N-P_{0}^N)  \\
                                   &\qquad\le \frac{1}{2} + \frac{1}{2} \TV(P_0^N, P_{1}^N)  \label{eq:use_villani} \\
                                   &\qquad\le \frac{1}{2} + \frac{N}{2} \TV(P_0, P_{1})  \label{eq:use_prod} \\
                                   &\qquad\le \frac{1}{2} + \frac{\gamma N}{2} \TV(P_0, P_{\backdoor}) \label{eq:use_convex} ,
  \end{align}
  where \cref{eq:use_villani} is a consequence of \cite[Exercise~1.17]{Villani2021Topics}. Also using \cite[Exercise~1.17]{Villani2021Topics}, we see that equality in \cref{eq:use_villani} is achieved for the Neyman-Pearson detector
  \begin{align}
    g_3(\DD, P_0, P_{\backdoor}) = \ind\left\{\frac{d P_{1}^N}{d P_{0}^N}(\DD) \ge 1\right\} .
  \end{align}
  The last two steps \cref{eq:use_prod,eq:use_convex} follow from \cref{lem:TV}.
\end{proofE}

Before we can analyze detectors of \cref{itm:hard,itm:harder}, we need to specify the set of allowable distributions $\mathcal P$. We do this, using \cref{lem:easier}.

First, we show that merely excluding the identity $P_0 \neq P_{\backdoor}$, i.e., $\mathcal P = \{(P_0, P_{\backdoor}) \in \mathcal P(\mathcal X)^2: P_0 \neq P_{\backdoor}\}$ is not sufficient.
\begin{example}
  \label{ex:ineq}
  Let $g_3(\DD,P_0,P_1)$ be an $\alpha$-error \cref{itm:easier} detector and assume that $\mathcal X$ is infinite, i.e., $|\mathcal X| = \infty$. Let $\mathcal P$ be given as above, ensuring only that $P_0 \neq P_{\backdoor}$. For any $\varepsilon > 0$, we can then choose\footnote{Without loss of generality, we can assume $\mathcal X = \mathbb N$. Then, this can, e.g., be achieved by $P_0 = \mathcal U(\{0,1,2,\dots,\lfloor \frac{\gamma N}{2\varepsilon} \rfloor\})$ and $P_{\backdoor} = \mathcal U(\{1,2,\dots,\lfloor \frac{\gamma N}{2\varepsilon} \rfloor\})$. We use $\mathcal U(\cdot)$ to denote a uniform distribution on a finite set.} $(P_0, P_{\backdoor}) \in \mathcal P$ with $0 < \TV(P_0, P_{\backdoor}) \le \frac{2}{\gamma N}\varepsilon$.
  By \cref{lem:easier}, we have
  $\alpha \ge \frac{1}{2} - \frac{\gamma N}{2} \TV(P_0, P_{\backdoor}) \ge \frac{1}{2} - \varepsilon$.
  As $\varepsilon > 0$ was arbitrary, we have $\alpha = \frac12$.
\end{example}
\Cref{lem:easier,ex:ineq} show that even for a \cref{itm:easier} detector, we need $\TV(P_0, P_{\backdoor}) > \frac{1-2\alpha}{\gamma N}$ for all $(P_0, P_{\backdoor}) \in \mathcal P$, in order for $\alpha$-error detection to be achievable.
In the following we will assume that $\mathcal P$ is the set of probability distributions $P_0, P_{\backdoor}$ with $\TV(P_0, P_{\backdoor}) \ge 1-\beta$, for some fixed $\beta \in [0,1)$. This strong requirement is motivated by the fact that in this case, $\frac{1-\gamma+\gamma\beta}{2}$-error \cref{itm:easier} detection is achievable with only one sample. 

\begin{remark}
  Thorough reasoning and examples, illustrating why total variation distance is the preferred distance measure for distribution hypothesis testing can be found in \cite[Section~1.2]{Canonne2022Topics}.
\end{remark}

\subsubsection{Impossibility}
\label{sec:impossibility}
In the following we prove an impossibility result, which implies that \emph{for an infinite alphabet $\mathcal X$, the error probability (as given in \cref{def:alpha-error}) of any detector (of \cref{itm:default}, \cref{itm:harder} or \cref{itm:hard}) is $\frac{1}{2}$, the error probability of a random guess.}
Additionally, for finite $\mathcal X$, we provide a lower bound on the size of the training set $N$, as a function of $\alpha$.

\begin{theoremE}
  \label{thm:imposs}
  Fix $N \in \mathbb N$, $\alpha \in (0, \frac12]$, $\beta \in [0,1]$, and $\mathcal P = \{(P_0, P_{\backdoor}) : \TV(P_0, P_{\backdoor}) \ge 1-\beta\}$. Let $g_2(\DD, P_0)$ be an $\alpha$-error \cref{itm:hard} detector.
  For $|\mathcal X| = \infty$, we then have necessarily $\alpha = \frac 12$, while for $|\mathcal X| < \infty$, we have
  \begin{align}
    \label{eq:imposs:condition}
        N &\ge \frac{\log 2\alpha}{2} + \sqrt{\frac{(\log 2\alpha)^2}{4} + (\beta |\mathcal X| - 1) \log\frac{1}{2\alpha} } .
  \end{align}
\end{theoremE}
\begin{proofE}
  For brevity we assume $P_0$ to be given and drop it as an argument for $g_2(\DD, P_0)$ = $g_2(\DD)$. Assume that $g_2$ is an $\alpha$-error detector.
  Without loss of generality, we will assume $|\mathcal X| = K \in \mathbb N$ and set $\mathcal X = \{1,\dots,K\}$. The case $|\mathcal X| = \infty$ will follow by letting $K \to \infty$.
  
  Choose $P_0 = \mathcal U(\mathcal X)$, the uniform distribution on $\mathcal X = \{1,\dots,K\}$. For an arbitrary, vector $\mathbf y = (y_1, y_2, \dots, y_M) \in \mathcal X^M$, let $\mathcal Q_{\mathbf y}$ be the discrete uniform
  distribution on the elements of $\mathbf y$. Note that this is only the uniform distribution on the set $\{y_m : m = 1,\dots, M\}$ if all components of $\mathbf y$ are different. Clearly, we have $\TV(P_0, \mathcal Q_{\mathbf y}) \ge 1-\frac MK$.
  Thus, by choosing $M \le \beta K$
  it is ensured that $\TV(P_0, \mathcal Q_{\mathbf y}) \ge 1-\beta$.

  Let $\mathbf Y = (Y_1, Y_2, \dots Y_M)$ be a random vector with $M$ elements, each drawn i.i.d.\ according to $Y_m \sim P_0$. We now draw another random vector $\mathbf Z$ with $N$ elements ${\mathbf Z} = \{Z_n\}_{n=1,2,\dots,N}$ according to $Z_n = (1-G_n) X^{(0)}_n + G_n Y_{V_n}$, where $V_n \sim \mathcal U(\{1,2,\dots,M\})$ and $G_n \sim \mathcal B(\gamma)$ are all independently drawn for $n=1,2,\dots,N$. Thus,  $V_n$ is uniformly drawn from $\{1,2,\dots,M\}$ and $G_n$ satisfies $\Pr\{G_n = 1\} = \gamma$ and $\Pr\{G_n = 0\} = 1-\gamma$.

  We note the following two facts about this construction:
  \begin{enumerate}
  \item The marginal distribution of every $Z_n \in \mathbf Z$ is $P_0$, but the selection is non-i.i.d.\ as $Z_n$ and $Z_{n'}$ depend on each other through $\mathbf Y$. However, when conditioning on the fact that all components of $\mathbf V = (V_1, V_2, \dots, V_N)$ are pairwise distinct, then the random variables $Y_{V_n}$ and $Y_{V_{n'}}$ are independent for $n \neq n'$ and thus $\mathbf Z$ is a vector of i.i.d.\ variables distributed according to $P_0$.
  \item When conditioning on $\mathbf Y = \mathbf y$, we have a different situation, where $Z_n \sim (1-\gamma)P_0 + \gamma \mathcal Q_{\mathbf y}$ are i.i.d., and by choosing $M \le \beta K$, we have $(P_0, \mathcal Q_{\mathbf y}) \in \mathcal P$.
  \end{enumerate}

  Let $|\mathbf V| = |\{V_1, V_2, \dots, V_N\}| = N$ be the event that $\mathbf V$ contains pairwise distinct elements, i.e., no repetitions occur. Using the first fact above, we calculate
  \begin{align}
    &\Pr\{g_2(\mathbf Z) = 1\} \\
    &\qquad\le \Pr\Big\{g_2(\mathbf Z) = 1 \Big| |\mathbf V| = N\Big\} + \Pr\{|\mathbf V| \neq N \} \nonumber\\
    &\qquad\le \Pr\Big\{g_2(\mathbf Z) = 1 \Big|  |\mathbf V| = N \Big\} + 1 - \frac{M!}{M^N(M-N)!}  \nonumber\\
    &\qquad\le \Pr\Big\{g_2(\mathbf Z) = 1 \Big| |\mathbf V| = N \Big\} + 1 - \left(1-\frac{N}{M}\right)^N  \nonumber\\
    &\qquad= \Pr\Big\{g_2(\DD^{(0)}) = 1 \Big\} + 1 - \left(1-\frac{N}{M}\right)^N \\
    &\qquad= 2 - \Pr\{g_2(\DD^{(0)}) = 0\}  - \left(1-\frac{N}{M}\right)^N \\
    &\qquad\le 2-\Pr\{g_2(\DD^{(0)}) = 0\}  - \exp \frac{-N^2}{M-N}  , \label{eq:j1_bound}
  \end{align}
  where we used the union bound as well as the inequality $\log(1+x) \ge \frac{x}{1+x}$.

  Using the second fact from above, we condition on $\mathbf Y = \mathbf y$ and then have $\mathbf Z$ i.i.d.\ according to $P_1 = (1-\gamma)P_0 + \gamma P_{\backdoor}$ for a valid backdoor distribution $P_{\backdoor} = \mathcal Q_{\mathbf y}$. We then write
  \begin{align}
    &\frac{1}{2} \Pr\{g_2(\DD^{(0)}) = 0\} + \frac{1}{2} \Pr\{g_2(\mathbf Z) = 1\} \\
    &\qquad= \frac{1}{2} \Pr\{g_2(\DD^{(0)}) = 0\} + \frac{1}{2} K^{-M} \sum_{\mathbf y \in \mathcal X^M} \Pr\big\{g_2(\mathbf Z) = 1 \big| \mathbf Y = \mathbf y \big\} \\
    &\qquad= K^{-M} \sum_{\mathbf y \in \mathcal X^M} \bigg( \frac{1}{2} \Pr\big\{g_2(\DD^{(0)}) = 0\big\} + \frac{1}{2}  \Pr\big\{g_2(\mathbf Z) = 1 \big| \mathbf Y = \mathbf y \big\} \bigg) \\
    &\qquad= K^{-M} \sum_{\mathbf y \in \mathcal X^M} \bigg( \frac{1}{2} \Pr\big\{g_2(\DD^{(0)}) = 0\big\} + \frac{1}{2}  \Pr\big\{g_2(\DD^{(1)}) = 1 \big| \mathbf Y = \mathbf y \big\} \bigg) \\
    &\qquad\ge K^{-M} \sum_{\mathbf y \in \mathcal X^M} (1-\alpha) \label{eq:alpha_bound} \\
    &\qquad=  1-\alpha . \label{eq:alpha_bound2}
  \end{align}
  In total we have
  \begin{align}
    1-\alpha&\stackrel{\mathclap{\labelcref{{eq:alpha_bound2}}}}{\le} \frac{1}{2} \Pr\{g_2(\DD^{(0)}) = 0\} + \frac{1}{2} \Pr\{g_2(\mathbf Z) = 1\} \\
            &\stackrel{\mathclap{\labelcref{{eq:j1_bound}}}}{\le} \frac{1}{2} \bigg(\Pr\{g_2(\DD^{(0)}) = 0\} + 2 -\Pr\{g_2(\DD^{(0)}) = 0\} - \exp \frac{-N^2}{M-N} \bigg) \\
            &=  1 -  \frac{1}{2} \exp \frac{-N^2}{M-N}
  \end{align}
  and thus
  \begin{align}
    \alpha &\ge \frac{1}{2}\exp \frac{-N^2}{M-N} .
  \end{align}

  This already resolves the case $|\mathcal X| = \infty$ as we can then let $K \to \infty$ and $M=\lfloor \beta K \rfloor \to \infty$, showing that $\alpha = \frac{1}{2}$ for $|\mathcal X| = \infty$.

  On the other hand, for $|\mathcal X| < \infty$, we choose $K = |\mathcal X|$, $M = \lfloor \beta K \rfloor$ and obtain \cref{eq:imposs:condition} by
  \begin{align}
    \alpha &\ge \frac{1}{2}\exp \frac{-N^2}{M-N} \\
    -\log 2\alpha &\le \frac{N^2}{\lfloor \beta K \rfloor -N}  \\
    0  &\le N^2 - N \log 2\alpha + \lfloor \beta K \rfloor \log 2\alpha  \\
    N &\ge \frac{\log 2\alpha}{2} + \sqrt{\frac{(\log 2\alpha)^2}{4} - \lfloor \beta K \rfloor \log 2\alpha } \\
    N &\ge \frac{\log 2\alpha}{2} + \sqrt{\frac{(\log 2\alpha)^2}{4} + (\beta K - 1) \log\frac{1}{2\alpha} }
  \end{align}
\end{proofE}

It is important to notice that the bound~\cref{eq:imposs:condition} relates the number of training samples $N$ with the alphabet size $|\mathcal X|$ and the risk $\alpha$, while the number $M$ of clean samples available to the defender does not appear. By the reduction argument in~\cref{lem:easier} the impossibility result in~\cref{thm:imposs} also holds for detector \cref{itm:harder,itm:default} for all possible values $M \in \mathbb N$.

For a fixed dataset alphabet size $|\mathcal X|$ and allowed error probability $\alpha$, the bound \cref{eq:imposs:condition} gives the minimum size of the training set $N$ for the error level $\alpha$ to be achievable.
Note the following special cases in terms of $\alpha$, $\beta$:
\begin{itemize}
\item For $\alpha = \frac 1 2$, the bound \cref{eq:imposs:condition} is always satisfied as the RHS is $0$, showing that $\frac12$-error detection is always achievable. This coincides with the error probability of a random guess.
\item The bound \cref{eq:imposs:condition} is monotonically decreasing in $\alpha$ and for $\alpha \to 0$, it approaches $\beta |\mathcal X|$.
\item In case $\beta = 0$, the bound \cref{eq:imposs:condition} is always satisfied as the RHS is zero for $\alpha \in (0,\frac12]$ in this case. This shows that $\alpha$-error detection is always possible if $P_0$ and $P_{\backdoor}$ have disjoint support, i.e., $\TV(P_0, P_{\backdoor}) = 1$.
\end{itemize}

For an infinite alphabet $\mathcal X$, \cref{eq:imposs:condition} needs to be satisfied for arbitrarily large values of $|\mathcal X|$. For finite training set size $N$, this is only possible if $\alpha = \frac12$ as then, $\log\frac{1}{2\alpha} = 0$. Thus, in this case, for any \cref{itm:hard} detector, there is a particular clean distribution and backdoor strategy, such that this detector performs no better than random guessing.

For fixed $\alpha$ and $\beta$, we can use \cref{eq:imposs:condition} to determine the minimum size of the training set $N$ for popular datasets, for $\alpha$ error probability to be achievable by a \cref{itm:hard} detector. To this end, we use the width $W$, height $H$, number of channels $C$ and color depth $D$ of an image dataset to compute $|\mathcal X| = D^{WHC}$. For categorical datasets, we may multiply the number of categories for all the properties recorded in the dataset to obtain $|\mathcal X|$. The resulting value for the bound in \cref{eq:imposs:condition} is given in \cref{tab:set_sizes} for several popular datasets.
As can be seen by these numbers, this universal backdoor detection is infeasible for all, but the smallest tabular datasets.
Code for computing the values in \cref{tab:set_sizes} can be found at \github.

Note also, that the impossibility of \cref{itm:hard} backdoor detection automatically precludes the existence of \cref{itm:harder} or \cref{itm:default} error detectors with equal performance by the reduction argument in \cref{rmk:ordering}.

\subsubsection*{Illustrative Example}
\label{sec:toy_example}

We noted previously the following consequence of \cref{thm:imposs}, in case of an infinite alphabet: For any \cref{itm:hard} backdoor detector, there exists an attacker, such that the detector is no better than a random guess. Here, we will showcase this on a toy example of a binary classification task, for a specific data distribution $P_0$, and any backdoor detector from a family of \cref{itm:hard} detectors, parameterized by $\mathbf v \in \mathbb R^K$.
For any parameter $\mathbf v$, we show how to construct a backdoor attack that is both effective in changing the decision regions of a classifier trained on backdoored data, and undetectable by the backdoor detector.

\textbf{Data Distribution.} We have data and label pairs $X = (Y, \mathbf Z)$ , where $Y \in \{-1,1\}$ is a binary label and $\mathbf Z = Y\mathbf 1 + \sigma \mathbf W$ with $\sigma > 0$ and $\mathbf W$ is multivariate normal with dimension $K$. For $K=2$ dimensions, the optimal classifier for this problem decides $\hat y = 1$ if $z_1 + z_2 \ge 0$ and otherwise $\hat y = -1$, leading to the decision boundary $z_2 = -z_1$.

\textbf{Backdoor Detector.} The \cref{itm:hard} detector $g_2(\DD, P_0)$ is parameterized by the unit vector $\mathbf v \in \mathbb R^K$ with $\Vert \mathbf v \Vert = 1$. Using the fact that
$f(X) := \mathbf v \cdot (Y\mathbf Z)
= \mathbf v \cdot \mathbf 1 + \sigma W$ with a standard normal variable $W$, we see that applying the function $f({}\cdot{})$ to $\DD$ yields $N$ i.d.d.\ Gaussian random variables with mean $\mu := \mathbf 1 \cdot \mathbf v$ and variance $\sigma^2$. The detector then performs a statistical goodness-of-fit test on this dataset. We utilize the Kolmogorov-Smirnov test for this purpose.

\textbf{Backdoor.} Knowing $\mathbf v$, the attacker transforms an input sample $x = (\mathbf z, y)$ into a backdoored sample $b(x) := (\mathbf z + y\Delta, -y)$ with the opposite label, and shifted by $y\Delta$, where $\Delta = \frac{2}{\sqrt{K-\mu^2}}(\mathbf 1 - \mu \mathbf v) - 2\mu\mathbf v$. This transformation ensures that the statistics of $f(X)$ do not change when applying the backdoor.

After the attacker replaces clean samples with backdoored samples at a rate of $\gamma \in (0,1]$, the Kolmogorov-Smirnov test is performed. \Cref{fig:impossibility_example} showcases this strategy in $K=2$ dimensions with $N=\myN$, $\gamma=\mygamma$, $\sigma=\mysigma$, and $\mathbf v = \myv$, resulting in $\mu = \mymu$. The Kolmogorov-Smirnov test obtained a p-value of $p_{\text{val}}=\mypval$, thus not detecting the backdoor. The resulting histograms of $f(\DD)$ for clean and backdoored data are shown in \cref{fig:impossibility_hist}. The code for this example can be found at \github.

\begin{figure}[h]
  \centering
  \input{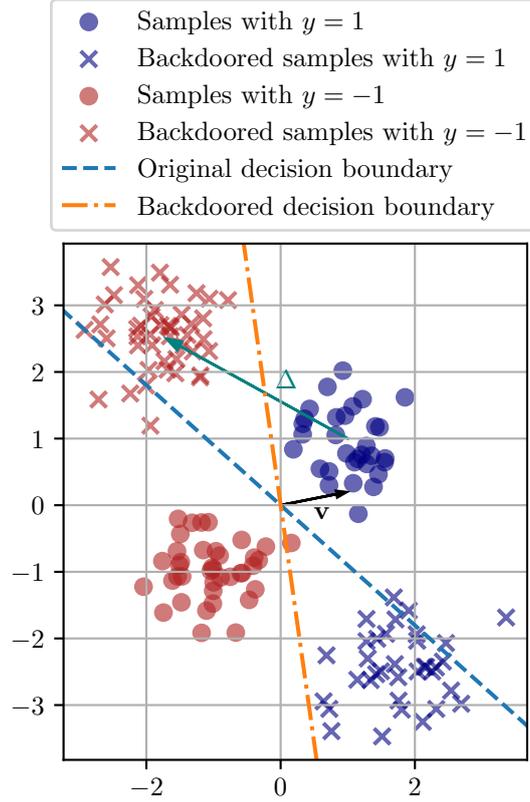}
  \caption{\Gls{mbd} example with $N=\myN$ samples. The backdoor detector uses projection onto $\mathbf v$ to take a decision. The vector $\Delta$ is the additive backdoor trigger used by the attacker. The decision boundary changes when applying the backdoor.}
  \label{fig:impossibility_example}
\end{figure}

\begin{figure}[h]
  \centering
  \input{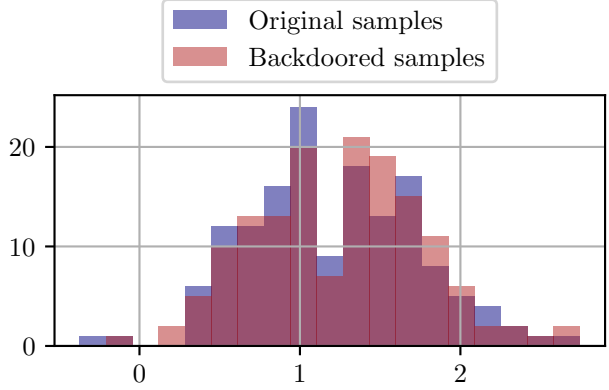}
  \caption{Histogram of the detector decision statistics clean and backdoored samples depicted in \cref{fig:impossibility_example}.}
  \label{fig:impossibility_hist}
\end{figure}


\subsubsection{Achievability}
\label{sec:achievability}
In this \lcnamecref{sec:achievability} we are going to show that $\alpha$-error \cref{itm:hard} detection is always achievable if the size of the alphabet $|\mathcal X|$ is small enough:
\begin{theoremE}
  \label{thm:achievability}
  Considering the backdoor detection setup of \cref{def:alpha-error} with $\mathcal P = \{(P_0, P_{\backdoor}) : \TV(P_0, P_{\backdoor}) \ge 1-\beta\}$ and a finite alphabet $|\mathcal X| < \infty$. If
  \begin{align}
    \label{eq:achievability}
    \alpha > 2|\mathcal X|\exp\left(-\frac{2N \gamma^2 (1-\beta)^2}{|\mathcal{X}|^2}\right) ,
  \end{align}
  then there exists an $\alpha$-error \cref{itm:hard} detector.
\end{theoremE}
\textEnd{In the proof of this \lcnamecref{thm:achievability}, the auxiliary \cref{lem:TV,lem:Pr} are used, which are provided in \cref{sec:auxiliary-results}.}
\begin{proofE}

  In the following we will show that the detector

  \begin{align} \label{eq:achievability:g}
    g(\DD, P_0) =
    \begin{cases}
      1 & \TV(P_0, S_N) \ge \gamma \frac{1-\beta}{2} \\
      0 & \text{otherwise}
    \end{cases}
  \end{align}
  is $\alpha$-error if \cref{eq:achievability} is satisfied. Here,
  the distribution $S_N$ is the so-called \emph{type} of $\DD$, i.e.,
  \begin{align}
    \label{eq:achievability:SN}
    S_N(x) = \frac{1}{N} \sum_{n=1}^N \indicator_x(X_n) ,
  \end{align}
  where for any $x\in\mathcal X$, $\indicator_{x}(X_n)$ is the indicator function that takes value $1$ if $X_n = x$ and $0$ otherwise.
  
  In \cref{lem:Pr} it is shown that the type $S_N$ is close to the true distribution $P$ with high probability.
  We can now analyze the error probability of the detector \cref{eq:achievability:g} for $P=P_1$, i.e.,
  \begin{align}
    \Pr\{g(\DD^{(1)}, P_0) = 0\} &= \Pr \left\{\TV(S_N^{(1)}, P_0) \le \gamma \frac{1-\beta}{2} \right\} \\
                                           &\le \Pr \left\{\TV(S_N^{(1)}, P_0) \le \frac{\TV(P_0, P_1)}{2} \right\} \label{eq:achievability:gamma1} \\
                                           &\le \Pr\left\{\TV(S_N^{(1)}, P_1) \ge \frac{\TV(P_0, P_1)}{2}\right\}\\
                                           &\le \Pr\left\{\TV(S_N^{(1)}, P_1) \ge \gamma \frac{1-\beta}{2}\right\} \label{eq:achievability:gamma2} \\
                                           &\le 2|\mathcal X|\exp\left(-\frac{2N \gamma^2 (1-\beta)^2}{|\mathcal{X}|^2}\right) , \label{eq:achievability:P0}
  \end{align}
  where we used \cref{lem:Pr} in \cref{eq:achievability:P0} and the fact that $\TV(P_0, P_1) = \gamma\TV(P_0, P_{\backdoor}) \ge (1-\beta)\gamma$ by \cref{lem:TV} in \cref{eq:achievability:gamma1,eq:achievability:gamma2}.
  Similarly, we obtain that the error probability for $j=0$ is upper bounded by the same expression
  \begin{align}
    \Pr\{g(\DD^{(0)}, P_0) = 1\} &= \Pr \left\{\TV(S_N^{(0)}, P_0) \ge \gamma \frac{1-\beta}{2} \right\} \\
                                           &\le 2|\mathcal X|\exp\left(-\frac{2N \gamma^2 (1-\beta)^2}{|\mathcal{X}|^2}\right) , \label{eq:achievability:P1} 
  \end{align}
  applying \cref{lem:Pr} in \cref{eq:achievability:P1}.

  Thus, we have shown that $g$, as defined in \cref{eq:achievability:g}, is $\alpha$-error, provided that \cref{eq:achievability} holds.
\end{proofE}


Note the following special cases 
for $\alpha$, $\beta$, $\gamma$, and $|\mathcal X|$:
\begin{itemize}
\item The bound on the RHS of \cref{eq:achievability} increases monotonically from $0$ to $\infty$ for increasing $|\mathcal X|$. Thus, there is some fixed alphabet size, below which, $\alpha$-error detection is guaranteed to be possible.
\item For $\alpha = 0$, \cref{eq:achievability} cannot be satisfied.
\item The case $\beta = 1$ allows for $P_0 = P_{\backdoor}$ and thus no $\alpha$-error detector exists for $\alpha \in [0,\frac12)$ in this case and \cref{eq:achievability} cannot be satisfied.
\item For $\gamma = 0$, the distributions $P_0 = P_{\backdoor}$ are identical and thus no $\alpha$-error detector exists for $\alpha \in [0,\frac12)$ in this case and \cref{eq:achievability} cannot be satisfied.
\end{itemize}

\subsubsection{Connections to \glsentryshort{pac}-Learnability of \glsentryshort{ood} Detection}
\label{sec:pac}

\begin{table}
  \centering
  \caption{Lower bound \cref{eq:imposs:condition} on $N$ evaluated for popular datasets with $\alpha = \myalpha$ and $\beta = \mybeta$.}
  \label{tab:set_sizes}
  \begin{tabular}{lll}
\toprule
Dataset & $|\mathcal{X}|$ & N \\
\midrule
Lisa Traffic Sign & $256^{307200}$ & $ \ge 10^{369904}$ \\
ImageNet & $256^{150528}$ & $ \ge 10^{181252}$ \\
CIFAR10 & $256^{3072}$ & $ \ge 10^{3697}$ \\
MNIST & $256^{784}$ & $ \ge 10^{942}$ \\
B/W MNIST & $2^{784}$ & $ \ge 10^{116}$ \\
Adult & $\ge 10^{21.86}$ & $ \ge 10^{9}$ \\
Heart Disease & $\ge 10^{13.51}$ & $ \ge 10^{5}$ \\
Iris & $\ge 10^{6.35}$ & $ \ge 10^{1}$ \\
\bottomrule
\end{tabular}

\end{table}

Note that a \cref{itm:harder} detector essentially needs to solve an \gls{ood} detection problem: The detector $g_1$ needs to determine if the $N$ samples $\DD$ were drawn from the same distribution as $\DDi$.

The goal of this \lcnamecref{sec:pac} is to prove \cref{thm:ood_connection}. This \lcnamecref{thm:ood_connection} has an interesting implication in case the \gls{ood} detection problem is \gls{pac}-learnable: If an $\alpha$-error \cref{itm:easier} backdoor detector $g_3$ exists, then $(\alpha+\epsilon)$-error detection is also possible for a \cref{itm:harder} detector for any $\epsilon > 0$.
Thus, essentially \cref{itm:harder,itm:hard,itm:easier} all become equivalent if \gls{ood} detection is \gls{pac}-learnable.
Note that \cref{itm:easier} detection is characterized by \cref{lem:easier}.

The \gls{pac}-learnability of the detector in \cref{itm:harder} was analyzed in \cite{Fang2022Is}.
We fist restate a special case of the definition of (weak) \gls{pac}-learnability as given in \cite[Def.~1]{Fang2022Is}.

\begin{definition}
  \label{def:pac-learnable}
  For distributions $P_0, P_{\backdoor}$ on $\mathcal X$, the \emph{\glsentryshort{ood}-risk} of a function $f : \mathcal X \to \{0,1\}$, w.r.t.\ the Hamming distance, is defined as
  \begin{align*}
    &\bar{R}(f, P_0, P_{\backdoor}) := \Pr\{f(X^{(J)}) \neq J\} \\
                                   &\qquad= \frac12 \Pr\{f(X^{(0)}) = 1\} + \frac12 \Pr\{f(X^{(1)}) = 0\} .
  \end{align*}
  
  Given a space of probability function $\mathcal P$, \emph{\glsentryshort{ood}-detection is \glsentryshort{pac}-learnable} on $\mathcal P$ if there exists an algorithm $\mathcal G \colon \bigcup_{m=1}^\infty \mathcal X^m \to \{0,1\}^{\mathcal X}$ and a monotonically decreasing sequence $\epsilon(m)$ such that $\lim_{m\to\infty} \epsilon(m) = 0$ and for all $(P_0, P_{\backdoor}) \in \mathcal P$, and all $M \in \mathbb N$ we have\footnote{Note that $\DDi$ contains $M$ samples.}
  \begin{align}
    \label{eq:ood-learnable}
    \Exp[\bar{R}(\mathcal G(\DDi), P_0, P_{\backdoor})] - \inf_{f} \bar{R}(f, P_0, P_{\backdoor}) \le \epsilon(M) ,
  \end{align}
  where the expectation is taken w.r.t.\ $\DDi$ and the infimum is over $\{0,1\}^{\mathcal X}$, i.e., all functions $f\colon \mathcal X \to \{0,1\}$.
\end{definition}
\begin{remark}
  \Cref{def:pac-learnable} is a special case of \cite[Def.~1]{Fang2022Is} in several ways\footnote{We use the notation of \cite[Sec.~2]{Fang2022Is} for the following symbols: $\mathcal H$, $X_O$, $X_I$, $Y_O$, $Y_I$, $D_{X_O}$, $D_{X_I}$, $D_{XY}$, $\mathscr D_{XY}$, $\ell(\cdot, \cdot)$, $K$.}:
  \begin{itemize}
  \item The hypothesis space is the complete function space $\mathcal H = \{0,1\}^{\mathcal X}$, of functions $f\colon \mathcal X \to \{0,1\}$.
  \item The loss function, as used in \cite[Eq.~(1)]{Fang2022Is} is the Hamming distance, i.e., $\ell(y, y') = 1$ if and only if $y\neq y'$.
  \item We are purely concerned with one-class novelty detection, i.e., $K=1$ in \cite[Sec.~2]{Fang2022Is}.
    Therefore we do not take $Y_O$ and $Y_I$ into account, as $Y_I \equiv 1$ and $Y_O \equiv 2$.
  \item Note that $(P_0, P_{\backdoor}) \in \mathcal P$ play the role of $(D_{X_O}, D_{X_I})$ and the complete domain space is then given by $\mathscr D_{XY} = \{D_{XY} : D_{XY} = \frac12 P_0 + \frac12 P_{\backdoor}, (P_0,P_{\backdoor}) \in \mathcal P\}$. 
  \end{itemize}
  Note, that strong \gls{pac}-learnability~\cite[Def.~2]{Fang2022Is} implies weak learnability.
\end{remark}

To connect \gls{pac}-learnability of \gls{ood} detection to the learning of backdoor detectors,  we consider \gls{pac}-learnability on the $N$-dimensional product space, i.e., on $\mathcal X^N$ with distributions $P_0^N, P_{\backdoor}^N$.

We can now connect \gls{pac}-learnability to the existence of $\alpha$-error detectors of \cref{itm:harder,itm:easier}.
\begin{theoremE}
  \label{thm:ood_connection}
  Consider the backdoor detection setup of \cref{def:alpha-error}, with fixed $\gamma \in (0,1]$, $N\in\mathbb N$ and some set of possible distributions $\mathcal P$.
  Let $\mathcal P'$ be the set of $N$-fold products of $(P_0, P_1)$, i.e., $\mathcal P' = \{(P_0^N, (\gamma P_{\backdoor} + (1-\gamma)P_0)^N) : (P_0, P_{\backdoor}) \in \mathcal P\}$.
  Then, \gls{ood}-detection is \gls{pac}-learnable on $\mathcal P'$ if and only if the following holds for any $\epsilon > 0$ and any \cref{itm:easier} detector $g_3(\DD, P_0, P_{\backdoor})$: We can find $M \in \mathbb N$ and a \cref{itm:harder} detector $g_1(\DD, \DDi)$, which satisfies $R(g_1, P_0, P_{\backdoor}) \le R(g_3, P_0, P_{\backdoor}) + \epsilon$ for every $(P_0, P_{\backdoor}) \in \mathcal P$.
\end{theoremE}
\textEnd{In the proof of this \lcnamecref{thm:ood_connection}, the auxiliary \cref{lem:ood_intermediate} is used, which is provided in \cref{sec:auxiliary-results}}
\begin{proofE}
  Assume first that \gls{ood}-detection is \gls{pac}-learnable on $\mathcal P'$, fix $\epsilon > 0$ and let $g_3$ be any \cref{itm:easier} detector.
  By \cref{lem:ood_intermediate}, we know that there is a \cref{itm:harder} detector $g_1^M$ with some $M$ such that $\epsilon(M) \le \epsilon$, satisfying \cref{eq:ood_intermediate}. Noting that $\frac12 - \frac12 \TV(P_0^N, P_{1}^N) \le R(g_3, P_0, P_{\backdoor})$ by \cref{lem:easier} completes this part of the proof.
  
  On the other hand, let $g_3$ be the \cref{itm:easier} Neyman-Pearson detector that satisfies $R(g_3, P_0, P_{\backdoor}) = \frac12 - \frac12 \TV(P_0^N, P_{1}^N)$, which exists by \cref{lem:easier}.
  By our assumptions, for any $k \in \mathbb N$, we set $\epsilon = \frac1k$ and find a \cref{itm:harder} detector $\hat g_1^k$, operating on $\DD$ with size $M=m(k)$ satisfying
  \begin{align}
    \frac12 \TV(P_0^N, P_{1}^N) - \frac12 + R(\hat g_1^k, P_0, P_{\backdoor})  \le \frac{1}{k} .
  \end{align}
  We can find a monotonically increasing sequence $k_M$ for $M=1,2,\dots$ with $\lim_{M\to \infty} k_M = \infty$, that satisfies $m(k_M) \le M$.
  Using the sequence of \cref{itm:harder} detectors\footnote{We use the notation $[\mathbf x]_k^l = [(x_1, x_2, \dots, x_N)]_k^l = (x_k, x_{k+1}, \dots, x_{l})$ for slicing.} $g_1^M(\DD, \DDi) = \hat g_1^{k_M}(\DD, [\DDi]_1^{m(k_M)})$ and $\epsilon(M) = \frac{1}{k_M}$, we have for every $(P_0, P_{\backdoor}) \in \mathcal P$,
  \begin{align}
    \epsilon(M) &\ge \frac12 \TV(P_0^N, P_{1}^N) - \frac12 + R(\hat g_1^{k_M}, P_0, P_{\backdoor}) \\
                &= \frac12 \TV(P_0^N, P_{1}^N) - \frac12 + R(g_1^{M}, P_0, P_{\backdoor}) .
  \end{align}
  This completes the proof as $\lim_{m \to \infty} \epsilon(m) = \lim_{m\to \infty} \frac{1}{k_m} = 0$ and thus, \gls{pac} learnability is guaranteed by \cref{lem:ood_intermediate}.
\end{proofE}

\begin{corollary}
  If \gls{ood}-detection is \gls{pac}-learnable on $\mathcal P'$, we have the following: If $\alpha$-error backdoor detection is possible in the easier case of \cref{itm:easier} detection, which is completely characterized by \cref{lem:easier}, then $(\alpha+\epsilon)$-error detection is also possible for a \cref{itm:harder} detector for any $\epsilon > 0$.

  Consequently,  up to topological closure, the same error probability is achievable for all detector \cref{itm:harder,itm:hard,itm:easier}, if \gls{ood}-detection is \gls{pac}-learnable on $\mathcal P'$.
\end{corollary}

\subsection{Generalizing to \glsentrylong{sbd}}
\label{sec:sbd_and_ood}

We can generalize \cref{def:alpha-error} to \gls{sbd} by providing a detector $g'(\mathbf Q')$ with input $\mathbf Q' = (\mathbf Q, X^{(I)}) = (\mathcal A(\DD^{(J)}), \DDi, X^{(I)})$, where a random variable $I$ on $\{0,1\}$ determines if a sample $X^{(I)}$ was drawn as $X^{(0)} \sim P_0$ ($I=0$) or as\footnote{Note, that $X^{(1)}$ is distributed according to $P_{\backdoor}$ and \textbf{not} according to $P_1 = (1-\gamma)P_0 + \gamma P_{\backdoor}$.} $X^{(1)} \sim P_{\backdoor}$ ($I=1$).

We define a general target function $t(j,i) \in \{0,1\}$ and require that a backdoor detector satisfies $g'(\mathbf Q') = t(J, I)$ with high probability.
In this case, it is beneficial to allow for an arbitrary probability distribution $P_{JI}$ of $(J, I)$ on $\{0,1\}^2$.
This naturally leads to the following alternative definition of $\alpha$-error detection, generalizing \cref{def:alpha-error}.
\begin{definition}
  \label{def:alpha-error1}
  A backdoor detection problem for a training algorithm $\mathcal A$ is determined by the following quantities:
  $\gamma \in (0,1]$, 
  $N \in \mathbb N$,
  $M \in \mathbb N$, 
  $\mathcal P \subseteq \mathcal P(\mathcal X)^2$,
  $P_{JI} \in \mathcal P(\{0,1\}^2)$, and
  $t\colon \{0,1\}^2 \to \{0,1\}$.

  Fixing these quantities, we define the \emph{risk} of a backdoor detector $g'$ associated with $(P_0, P_{\backdoor})$ as
  $R(g'; P_0, P_{\backdoor}) := \Pr\{g'(\mathbf{Q'}) \neq t(J,I)\}$,
  where the probability is w.r.t.\ $\mathbf{Q'} = (\mathcal A(\DD^{(J)}), \DDi, X^{(I)})$ and $(J, I)$.
  
  We say that a backdoor detector is \emph{$\alpha$-error} for some $\alpha \in [0,\frac 12]$ if, for every pair $(P_0, P_{\backdoor}) \in \mathcal P$, the risk is bounded by 
  $R(g'; P_0, P_{\backdoor}) \le \alpha$.    
\end{definition}

Then, even \gls{ood} can be modeled using our setup.
\Cref{fig:flavors} presents an overview of the target function $t(j,i)$ for \gls{mbd}, \gls{sbd} and \gls{ood}.

\tikzset{mymatrix/.style={
    column 1/.style={nodes={minimum width=12mm, draw=none,minimum height=0pt, text width=5mm}},
    row 1/.style={nodes={draw=none,minimum height=8mm}},
    nodes={draw, minimum height=10mm, outer sep=0pt, minimum width=10mm, inner sep=0pt, anchor=center,
      text width=10mm, align=center},
    matrix of nodes,
    inner sep=0pt,
    outer sep=0pt,
    column sep=-\pgflinewidth,
    row sep=-\pgflinewidth,
    CS/.style = {fill=gray!30},
  }}
\begin{figure*}[h]
  \centering
  {\hspace{5em}Sample Backdoored? ($i$)}
  
  \begin{minipage}[t][0cm][b]{.05\textwidth}
    \flushright
      \begin{tikzpicture}
      \node[rotate=90] {Model Backdoored? ($j$)};
    \end{tikzpicture}
  \end{minipage}
  \begin{subfigure}[b]{.25\textwidth}
    \centering
    \begin{tikzpicture}
      \matrix (m) [mymatrix]
      {       & $0/N$ & $1/Y$ \\
        $0/N$   & 0   & |[CS]| -- \\
        $1/Y$   & 1   & |[CS]| -- \\
      };
    \end{tikzpicture}
    \caption{$t_{\text{\glsentryshort{mbd}}}(j,i) = j$.}
    \label{fig:mdb}
  \end{subfigure}
  \begin{subfigure}[b]{.25\textwidth}
    \centering
    \begin{tikzpicture}
      \matrix (m) [mymatrix]
      {       & $0/N$ & $1/Y$ \\
        $0/N$   & 0   & |[CS]| -- \\
        $1/Y$   & 0   & 1 \\
      };
    \end{tikzpicture}
    \caption{$t_{\text{\glsentryshort{sbd}}}(j,i) = i$.}
    \label{fig:sdb}
  \end{subfigure}
  \begin{subfigure}[b]{.25\textwidth}
    \centering
    \begin{tikzpicture}
      \matrix (m) [mymatrix]
      {       & $0/N$ & $1/Y$ \\
        $0/N$   & 0   & 1 \\
        $1/Y$   & |[CS]| --   & |[CS]| -- \\
      };
    \end{tikzpicture}
    \caption{$t_{\text{\glsentryshort{ood}}}(j,i) = t_{\text{\glsentryshort{sbd}}}(j,i) = i$.}
    \label{fig:cdb}
  \end{subfigure}
  \caption{Target function $t(j,i)$ for different backdoor detection flavors. $j \in \{0,1\}$ signals if the training dataset is backdoored ($j=1$) or not ($j=0$), while $i \in \{0,1\}$ indicates if the test sample is backdoored.}
  \label{fig:flavors}
\end{figure*}
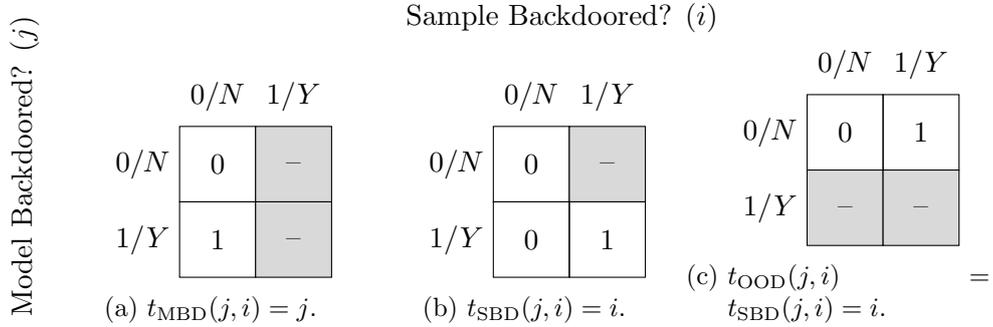

Note that several cells in the diagrams in \cref{fig:flavors} are grayed out. This reflects the fact that for certain flavors of backdoor detection, specific combinations of $(j,i)$ are not relevant. For \gls{mbd} for instance, we are not interested in whether the target sample $X^{(I)}$ contains a backdoor and we can thus assume $I=0$ in this case, effectively reducing this case to the problem introduced in \cref{sec:threat-model} with $M+1$ samples being drawn from $P_0$, i.e., $(\DDi, X^{(0)})$, available to the detector.
Conversely, the case of a clean model, i.e., $j=0$ and a sample with a backdoor, i.e., $i=1$ is not realistic for \gls{sbd} and we set $P_{JI}(0,1) = 0$ in this case.
By setting $J=0$ (i.e., model is trained on clean data and $P_{JI}(1,0) = P_{JI}(1,1) = 0$) and using $t_{\text{\glsentryshort{ood}}}(j,i) = t_{\text{\glsentryshort{sbd}}}(i,0) = i$, we obtain an \gls{ood} detection problem, where the detector has access to a model $\mathcal A(\DD^{(0)})$ trained on clean data and additional clean data $\DDi$. The detector then needs to determine whether $X^{(I)}$ is in or out of distribution.

To show how our result from \cref{sec:impossibility,sec:achievability} carry over to other variants of backdoor detection, we will directly use \cref{thm:imposs} to derive a similar result for \gls{sbd}.
In analogy to the different \namecrefs{itm:default} of \gls{mbd} detectors introduced in \cref{sec:results}, we have a \cref{itm:hard} detector $g_2'(\mathbf Q_2')$ with $\mathbf Q_2' = (\DD^{(J)}, P_0, X^{(I)})$ for \gls{sbd}.

For such a detector we can leverage a reduction argument to obtain the following.
\begin{corollaryE}
  \label{cor:sbd}
  Let $g_2'(\DD^{(J)}, P_0, X^{(I)})$ be a \cref{itm:hard} detector for an \gls{sbd} problem, where we have $r = \min\{P_{JI}(0,0), P_{JI}(1,1)\} > 0$ and $\mathcal P = \{(P_0, P_{\backdoor}) : \TV(P_0, P_{\backdoor}) \ge 1-\beta\}$. Then, if $g_2'$ is $\alpha$-error,
  we have $\alpha \ge r$ if $|\mathcal X| = \infty$,
  and for $|\mathcal X| < \infty$, we obtain
  \begin{align}
    \label{eq:sbd}
    N &\ge \frac{\log \frac{\alpha}{r}}{2} + \sqrt{\frac{(\log\frac{\alpha}{r})^2}{4} + (\beta |\mathcal X| - 1) \log\frac{r}{\alpha} }.
  \end{align}
\end{corollaryE}
\begin{proofE}

  Assuming that this detector is $\alpha$-error implies
  \begin{align}
    \alpha \ge R(g'_2, P_0, P_{\backdoor}) &\ge P_{JI}(0,0) \Pr\{g_2'(\mathbf Q_2') \neq 0 | J=I=0\} \nonumber\\*
                                           &\qquad+ P_{JI}(1,1) \Pr\{g_2'(\mathbf Q_2') \neq 1 | J=I=1\}  \\
                                           &\ge r \big(  \Pr\{g_2'(\mathbf Q_2') \neq 0 | J=I=0\} 
                                             +  \Pr\{g_2'(\mathbf Q_2') \neq 1 | J=I=1\} \big) .
  \end{align}

  Now consider the \gls{mbd} problem with $\gamma = 1$ and the training set size $N'=N+1$. We can define a \cref{itm:hard} detector\footnote{If $\gamma > 0$ for the \gls{sbd} problem, randomly replace elements of $\DD$ by independently drawn realizations of $P_0$.} $g_2(\DD, P_0) = g_2'(\DD, P_0, X_{N'})$ with risk
  \begin{align}
    R(g_2, P_0, P_{\backdoor}) &= \frac12 \Pr\{g_2'(\mathbf Q_2') \neq 0 | J=I=0\} + \frac12 \Pr\{g_2'(\mathbf Q_2') \neq 1 | J=I=1\} \\
                               &\le \frac{1}{2r} \alpha .
  \end{align}

  From \cref{thm:imposs}, we now know that $\frac{1}{2r} \alpha \ge \frac12$ if $\mathcal X = \mathbb N$ and obtain \cref{eq:sbd} for $|\mathcal X| < \infty$.
\end{proofE}

\section{\uppercase{Related Works}}
\label{sec:related_works}
\textbf{Backdoor attacks.} 
    Early backdoor methods rely on triggers that are visible to the human eye, and generally consist of a local patch on the samples \cite{Gu2017BadNets,Shafahi2018Poison,Nguyen2020Input}. Other attacks add a layer of stealthiness by using invisible triggers, which are commonly covering the whole sample and are not detectable by the human eye \cite{Chen2017Targeted,Zeng2021Rethinking,Li2021Invisible}. Additive attacks \cite{Gu2017BadNets,Chen2017Targeted,Shafahi2018Poison} fuse the triggers to the clean samples as additive noise. 
    Conversely, non-additive attacks \cite{Zeng2021Rethinking,Li2021Invisible,Nguyen2021WaNet} modify the samples by changing attributes such as the color of the pixels or applying spatial transformations. 
    Additionally, some attacks add the same trigger to all samples and are therefore sample-agnostic \cite{Gu2017BadNets,Chen2017Targeted}, while others are sample-specific \cite{Nguyen2020Input,Nguyen2021WaNet}. 
    Finally, recent attacks have been proposed to reduce the issue of the linear separability between clean and backdoored samples \cite{Qi2023Revisiting} which arises in many of the previously mentioned works. 

    \textbf{Backdoor defenses.} This work is concerned with ``post-training'' defenses, i.e. those methods that aim to remove or mitigate the backdoor effect from a backdoored model, as opposed to techniques that deal with the problem before~\cite{Udeshi2022Model,Gao2023Backdoor} or during training~\cite{Huang2022Backdoor}. The solutions that mostly align with the proposed frameworks are those that are designed to detect whether the model is backdoored \cite{Liu2017Neural,Wang2019Neural}, or those that can detect backdoored samples 
    \cite{Li2021Anti,Huang2022Backdoor,Liu2019ABS}.
    While some detector only work when the trigger is assumed to be sample-agnostic \cite{Chou2020SentiNet,Gao2019Strip,Tao2022Better}, others are reported to be effective on sample-specific triggers \cite{Zeng2021Rethinking,Liu2023Detecting}. Moreover, recent detectors propose to replace the need for a set of clean samples with the generation of perturbed samples which may help to create a representation of the backdoored samples \cite{Liu2023Detecting,Pang2023Backdoor}.
    In providing a theoretical analysis from first principles, and in suggesting a strong connection between \gls{ood} detection and backdoor detection, our work is complementary to \cite{Ma2022Beatrix}.
    Finally, is important to notice that most literature on backdoor detection focuses on the \gls{sbd} problem and on the mitigation of backdoor attacks, when a dataset is known to be backdoored~\cite{Udeshi2022Model,Huang2022Backdoor,Tran2018Spectral,Wu2021Adversarial}. 

    \textbf{Certifiable Defenses.} For data poisoning attacks, where the attacker's goal is only to diminish the accuracy of the trained classifier, certifiable defenses do exist~\cite{Steinhardt2017Certified,Koh2022Stronger}. Also for backdoor attacks, smoothing strategies~\cite{Weber2023Rab,Wang2020Certifying} were proposed, which allow for certified robustness against backdoor attacks. However, the threat model is severely restricted to very small (in size and amplitude) triggers, which can be successfully obscured by adding smoothing noise. 
    The impossibility result in \cref{sec:impossibility} showcases, why certifiable defenses cannot be mounted against a capable attacker in general.




\section{\uppercase{Conclusions}}
\label{sec:conclusions}

We provided a formal statistical definition of backdoor detection and investigated the feasibility of backdoor detection.
As the backdoor attack is usually not known to the defender,
in our analysis we focused on universal (adversary-unaware) backdoor detection.
This implies that such backdoor defense schemes must be robust against targeted attacks, which are crafted to fool the specific defense strategy,
excluding any ``security-by-obscurity'' schemes, where defense only holds as long as it is not public knowledge.
We concluded that under very general assumptions, universal (adversary-unaware) backdoor detection is not possible.
Thus, backdoor detectors need to be adversary-aware to perform well at their task.

Ultimately, this work makes the claim that designing universal (adversary-unaware) backdoor detection methods is an exercise in futility. As did \cite{Shokri2020Bypassing}, we make the case that backdoor detectors need to be adversary-aware or make specific assumptions on the data distribution and/or backdoor strategy employed. Unfortunately this is not the case for much published work, which implies that the proposed methods must fail on many untested instances of backdoor detection.

Furthermore, we note that when designing a backdoor detection algorithm, the advantage should be given to the attacker, which is able to adapt to a defense strategy, but not the other way around.

\subsubsection*{Acknowledgments}
This work is supported in part by grants from the National Science Foundation (NSF) and the ARO 77191, in part by the Army Research Office under grant \#W911NF-21-1-0155, in part by Intel, and in part by the NYUAD Center for Artificial Intelligence and Robotics, funded by Tamkeen under the NYUAD Research Institute Award CG010.

\bibliography{lit.bib} 



\clearpage
\appendix
\section{Proofs}
\label{sec:proofs}
\printProofs

\section{Auxiliary Results}
\label{sec:auxiliary-results}
This \lcnamecref{sec:auxiliary-results} contains auxiliary results, which are utilized in the proofs provided in \cref{sec:proofs}.

\begin{lemma}[Properties of Total Variation]
  \label{lem:TV}
  The total variation between two probability distributions $P_0, P_1 \in \mathcal P(\mathcal X)$, is given by
  \begin{align}
    \label{eq:TV}
    \TV(P_0, P_1) = \Vert P_0 - P_1 \Vert_{\mathrm{TV}} := \sup_{A} |P_0(A) - P_1(A)| ,
  \end{align}
  where the supremum is over all measurable sets $A \subseteq \mathcal X$. We then have
  \begin{align}
    \label{eq:TV_alt}
    \Vert P_0 - P_1 \Vert_{\mathrm{TV}} = 2 \inf_{X_0,X_1: P_{X_0} = P_0, P_{X_1} = P_1} \Pr\{ X_0 \neq X_1\} ,
  \end{align}
  where the infimum is over all random variables $X_0, X_1$ on $\mathcal X$, such that the marginal distributions
  satisfy $P_{X_0} = P_0$, $P_{X_1} = P_1$.
  For $P_0', P_1' \in \mathcal P(\mathcal Y)$, we have
  \begin{align}
    \label{eq:TV_prod}
    \Vert P_0 - P_1 \Vert_{\mathrm{TV}} \le \Vert P_0 \times P_0' - P_1 \times P_1' \Vert_{\mathrm{TV}} \le \Vert P_0 - P_1 \Vert_{\mathrm{TV}} + \Vert P_0' - P_1' \Vert_{\mathrm{TV}} .
  \end{align}
  and thus $\Vert P_0 - P_1 \Vert_{\mathrm{TV}} \le \Vert P_0^N - P_1^N \Vert_{\mathrm{TV}} \le N \Vert P_0 - P_1 \Vert_{\mathrm{TV}}$.
  Furthermore, for $\gamma \in [0,1]$,
  \begin{align}
    \label{eq:TV_convex}
    \Vert P_0 - (1-\gamma) P_0 - \gamma P_1 \Vert_{\mathrm{TV}} = \gamma \Vert P_0 - P_1 \Vert_{\mathrm{TV}} 
  \end{align}
\end{lemma}
\begin{proof}
  The characterization \cref{eq:TV_alt} can be found in \cite{Villani2021Topics}.

  To show the first inequality in \cref{eq:TV_prod}, observe that
  \begin{align}
    \Vert P_0 \times P_0' - P_1 \times P_1' \Vert_{\mathrm{TV}} &= \sup_{B} |[P_0 \times P_0'](B) - [P_1 \times P_1'](B)| \\
                                                                &\ge \sup_{A} |[P_0 \times P_0'](A \times \mathcal Y) - [P_1 \times P_1'](A \times \mathcal Y)| \\
                                                                &= \Vert P_0 - P_1 \Vert_{\mathrm{TV}} .
  \end{align}
  
  To show the second inequality in \cref{eq:TV_prod}, we use \cref{eq:TV_alt} and for an arbitrary $\varepsilon > 0$, choose $(X_0, X_1) \perp (Y_0, Y_1)$ such that
  $P_{X_0} = P_0$, $P_{X_1} = P_1$, $P_{Y_0} = P_0'$, $P_{Y_1} = P_1'$, and
  \begin{align}
    \Vert P_0 - P_1 \Vert_{\mathrm{TV}}+ \varepsilon &\ge 2 \Pr\{ X_0 \neq X_1\}  , \\
    \Vert P_0' - P_1' \Vert_{\mathrm{TV}} + \varepsilon &\ge 2 \Pr\{ Y_0 \neq Y_1\}  .
  \end{align}
  Clearly $P_{X_0, Y_0} = P_0 \times P_0'$ as well as $P_{X_1, Y_1} = P_1 \times P_1'$ and thus by \cref{eq:TV_alt},
  \begin{align}
    \Vert P_0 \times P_0' - P_1 \times P_1' \Vert_{\mathrm{TV}} &\le 2 \Pr\{ (X_0,Y_0) \neq (X_1,Y_1)\} \\
                                                                &\le 2 \Pr\{X_0 \neq X_1\} + 2\Pr\{Y_0 \neq Y_1\} \\
                                                                &\le \Vert P_0 - P_1 \Vert_{\mathrm{TV}} + \Vert P_0' - P_1' \Vert_{\mathrm{TV}} + 2\varepsilon .
  \end{align}
  As $\varepsilon > 0$ was arbitrary, this proves \cref{eq:TV_prod}.

  To show \cref{eq:TV_convex}, we use \cref{eq:TV} and have
  \begin{align}
    \label{eq:4}
    \Vert P_0 - (1-\gamma) P_0 - \gamma P_1 \Vert_{\mathrm{TV}} &= \sup_{A} |P_0(A) - (1-\gamma) P_0(A) - \gamma P_1(A)| \\
                                                                &= \sup_{A} |\gamma P_0(A) - \gamma P_1(A)| \\
                                                                &= \gamma \Vert P_0 - P_1 \Vert_{\mathrm{TV}} .
  \end{align}
\end{proof}

\begin{lemma}
  \label{lem:Pr}
  Let $S_N$ be the type of $\mathbf X = (X_1, X_2, \dots, X_N)$, distributed according to $P^N$. For any $t \in [0,1]$, we then have the bound
  \begin{align}
    \label{eq:Pr}
    \Pr\left\{\TV(S_N, P) \ge t \right\} &\le 2|\mathcal X|\exp\left(-\frac{8Nt^2}{|\mathcal{X}|^2}\right) .
  \end{align}
\end{lemma}
\begin{proof}
  By using the Hoeffding's inequality we can bound the probability of the deviation of $S_N$ from its expected value. In particular, we have that

  \begin{align}
    \label{eq:achievability:2}
    \Pr\left\{\left|S_N(x) - P(x)\right| \ge t\right\}  &=\Pr\left\{\left|S_N(x) - \mathbb E[S_N(x)]\right| \ge t\right\} \\
                                                        &\le 2\exp\left(\frac{-2t^2}{\sum_{n=1}^{N}(\frac 1 N - 0)^2}\right)\\
                                                        &= 2\exp\left(\frac{-2t^2}{\frac 1 N}\right) \\
                                                        &= 2\exp\left(-2Nt^2\right) ,
  \end{align}
  where we note that $\mathbb E[S_N(x)] = \frac1N \sum_{n=1}^N \mathbb E[\indicator_{x}[X_n]] = P(x)$.

  The next and final step is to extend the bound to the whole alphabet $\mathcal X$. In order to do so, we define the event $\mathcal{A}_{x} = \{\left|S_N(x) - P(x)\right| \ge t\}$.
  We want to bound the probability of the event
  \begin{align}
    \mathcal A = \bigcup_{x \in \mathcal X} \mathcal A_x  = \left\{ \exists x \in \mathcal X : \mathcal{A}_{x} \right\}.
  \end{align}

  By applying the union bound we obtain
  \begin{align}
    \Pr \mathcal A &= \Pr\left\{\bigcup_{x\in\mathcal X}\mathcal{A}_{x}\right\} \\
                   &\le \sum_{x\in\mathcal X}\Pr\left\{\mathcal{A}_{x}\right\} \\
                   &\le \sum_{x\in\mathcal X}2\exp\left(-2Nt^2\right)\\
    \label{eq:achievability:5}
                   &= 2|\mathcal X|\exp\left(-2Nt^2\right) .
  \end{align}

  Let us consider the event $\mathcal{A} = \{ \exists x \in \mathcal X : \left|S_N(x) - P(x)\right| \ge t\}$: this is the error event, i.e., the divergence between the observed samples frequency and its expected value diverges more than a given value $t>0$ for at least one $x \in \mathcal X$. The complement of this event is the event that the divergence is less than $t$ for all $x \in \mathcal X$, i.e., the event that the observed frequency is close to the expected value for all $x \in \mathcal X$. This can be written as
  \begin{align}
    \mathcal A^c = \{\forall x\in\mathcal{X},\ \left|S_N(x) - P(x)\right|<t \}.
  \end{align}
  Now, $\mathcal A^c$  implies that
  \begin{align}
    \sum_{x\in\mathcal X}\left|S_N(x) - P(x)\right| &< t|\mathcal X| \\
    \frac 1 2\sum_{x\in\mathcal X}\left|S_N(x) - P(x)\right| &<  \frac 1 2 t|\mathcal X| \\
    \TV(S_N, P) &< t^{\prime}
  \end{align}
  where $t^{\prime} = \frac 1 2 t|\mathcal X|$.
  Thus, $\Pr\mathcal A^c \le \Pr\{\TV(S_N, P) < t^{\prime}\}$ and therefore
  \begin{align} \label{eq:achievability:5.5}
    \Pr\{\TV(S_N, P) \ge t^{\prime}\} \le \Pr \mathcal A \le 2|\mathcal X|\exp\left(-2Nt^2\right) ,
  \end{align}
  where we have used \cref{eq:achievability:5}.
  By writing $t$ in terms of $t^{\prime}$ in \cref{eq:achievability:5.5}, we obtain \cref{eq:Pr}.
\end{proof}

\begin{lemma}
  \label{lem:ood_intermediate}
  Given $\mathcal P$ and $N \in \mathbb N$ and letting $\gamma \in (0,1]$, \gls{ood}-detection is \gls{pac}-learnable on $\mathcal P' = \{(P_0^N, P_{1}^N) : (P_0, P_{\backdoor}) \in \mathcal P\}$ with $P_1 = (1-\gamma)P_0 + \gamma P_{\backdoor}$ if and only if the following holds: For the \gls{mbd} problem, there exists a sequence of \cref{itm:harder} backdoor detectors $g_1^M(\DD, \DDi)$ for $M=1,2,\dots$ and a decreasing sequence $\epsilon(m)$ with $\lim_{m\to\infty} \epsilon(m) = 0$ such that for any $M \in \mathbb N$ and any pair $(P_0, P_{\backdoor}) \in \mathcal P$, we have
  \begin{align}
    \label{eq:ood_intermediate}
    \frac12 \TV(P_0^N, P_{1}^N) - \frac12 + R(g_1^M, P_0, P_{\backdoor})  \le \epsilon(M) .
  \end{align}
\end{lemma}
\begin{proof}
  Assume that \gls{ood}-detection is \gls{pac}-learnable on $\mathcal P'$. By definition we can find a function $\mathcal G \colon \bigcup_{m=1}^\infty \mathcal X^{Nm} \to \{0,1\}^{\mathcal X^N}$ and a monotonically decreasing sequence $\epsilon'(m)$ that tends to zero and satisfies for every $(P_0, P_{\backdoor}) \in \mathcal P$, $m \in \mathbb N$, that
  \begin{align}
    \label{eq:ood_eq}
    \Exp[\bar{R}(\mathcal G(\DDi), P_0^N, P_{1}^N)] - \inf_{f} \bar{R}(f, P_0^N, P_{1}^N) \le \epsilon'(m) ,
  \end{align}
  where we set the size of $\DDi$ to be $M = mN$ and the infimum is over all functions $f\colon \mathcal X^N \to \{0,1\}$.
  
  For any $M \in \mathbb N$, we define\footnote{We use the notation $[\mathbf x]_k^l = [(x_1, x_2, \dots, x_N)]_k^l = (x_k, x_{k+1}, \dots, x_{l})$ for slicing.}
  $g_1^M(\DD, \DDi) := \mathcal G([\DDi]_{1}^{mN})(\DD)$ as well as $\epsilon(M) = \epsilon'(m)$, where $m$ is the largest integer such that $mN \le M$.
  Notice that $R(g_1^M, P_0, P_{\backdoor}) = \Exp[\bar{R}(\mathcal G([\DDi]_{1}^{mN}), P_0^N, P_{1}^N)]$ and that $\inf_{f} \bar{R}(f, P_0^N, P_{1}^N) = \frac12 - \frac12 \TV(P_0^N, P_{1}^N)$ by \cref{lem:easier}. We thus obtain from \cref{eq:ood_eq}, that for any $M \in \mathbb N$,
  \begin{align}
    \epsilon(M) = \epsilon'(m) &\ge \Exp[\bar{R}(\mathcal G([\DDi]_{1}^{mN}), P_0^N, P_{1}^N)] - \frac12 + \frac12 \TV(P_0^N, P_{1}^N)   \\
                               &= R(g_1^M, P_0, P_{\backdoor}) - \frac12 + \frac12 \TV(P_0^N, P_{1}^N) .
  \end{align}
  Noting that $\epsilon(M)$ approaches zero completes this part of the proof.

  On the other hand, assume that $g_1^M(\DD, \DDi)$ and $\epsilon(M)$ satisfy the requirement \cref{eq:ood_intermediate}. For any $m \in \mathbb N$, we then set $M=mN$ and define $\mathcal G(\DDi)(\DD) := g_1^{M}(\DD, \DDi)$ as well as $\epsilon'(m) = \epsilon(M)$. We can now rewrite \cref{eq:ood_intermediate} using $\Exp[\bar{R}(\mathcal G(\DDi), P_0^N, P_{1}^N)] = R(g_1^M, P_0, P_{\backdoor})$ and \cref{lem:easier} to obtain
  \begin{align}
    \epsilon'(m) =\epsilon(mN) &\ge \frac12 \TV(P_0^N, P_{1}^N) - \frac12 + R(g_1^M, P_0, P_{\backdoor}) \\
                               &= \Exp[\bar{R}(\mathcal G(\DDi), P_0^N, P_{\backdoor}^N)] - \inf_{f} \bar{R}(f, P_0^N, P_{1}^N) .
  \end{align}
  Thus, we have shown that the algorithm $\mathcal G$ and the sequence $\epsilon'$ satisfy \cref{def:pac-learnable} and \gls{ood}-detection is \gls{pac}-learnable on $\mathcal P'$.
\end{proof}

\section{Additional Results}
\label{sec:additional-results}
\begin{lemma}
  \label{lem:reduction}
  Let $g_l$ be a detector as listed in \cref{sec:results} with input $\mathbf Q_l$ for $l \in \{0,1,2\}$, where we set $g_0 = g$ and $\mathbf Q_0 = \mathbf Q$. If $g_l$ is $\alpha$-error in the sense of \cref{def:alpha-error}, then for $m \in \{1,2,3\}$ and $m > l$ we can find a backdoor detector $g_m$ with input $\mathbf Q_m$ that is also $\alpha$-error.
\end{lemma}
\begin{proof}
  It is sufficient to show the \lcnamecref{lem:reduction} for $m=l+1$. The claim then follows by applying the result repeatedly.

  In the case $l=2$ (and $m=3$) we obtain $g_3$ with $R(g_3, P_0, P_{\backdoor}) = R(g_2, P_0, P_{\backdoor})$ by $g_3(\DD, P_0, P_{\backdoor}) = g_2(\DD, P_0)$.

  For $l=1$, we can define the randomized detector $g_2(\DD, x, P_0)$ to first draw $M$ i.i.d.\ samples $\DDi \sim P_0^M$ and then yield $g_2(\DD, P_0) = g_1(\DD, \DDi)$.

  Finally, for $l=0$ we obtain $g_1$ with equal risk by defining $g_1(\DD, \DDi) = g(\mathcal A(\DD), \DDi)$.
\end{proof}%



\end{document}